\documentclass{article}

\usepackage{graphicx}
\usepackage{latexsym}
\usepackage{fullpage}
\usepackage{amsfonts}
\usepackage{amsmath}
\usepackage{amsthm}
\usepackage{mathrsfs}
\usepackage{url}
\usepackage{algorithm}
\usepackage{algpseudocode}
\usepackage{charter,eulervm}
\usepackage{authblk}

\newcommand{\CW}{{\mathcal W}}
\newcommand{\mex}{{\rm mex}}
\newcommand{\nim}{{\rm nim}}
\newcommand{\mw}{{\rm mw}}
\newcommand{\nd}{{\rm nd}}

\newtheorem{theorem}{Theorem}
\newtheorem{lemma}[theorem]{Lemma}
\newtheorem{corollary}[theorem]{Corollary}
\newtheorem{definition}{Definition}

\algnewcommand{\IIf}[1]{\State\algorithmicif\ #1\ \algorithmicthen}
\algnewcommand{\EndIIf}{\unskip\ \algorithmicend\ \algorithmicif}

\begin{document}
\title{On Structural Parameterizations of Node Kayles}

\author{Yasuaki Kobayashi}
%\affil{Kyoto University}
\date{}
\maketitle 
\begin{abstract}
Node Kayles is a well-known two-player impartial game on graphs: Given an undirected graph, each player alternately chooses a vertex not adjacent to previously chosen vertices, and a player who cannot choose a new vertex loses the game.
  The problem of deciding if the first player has a winning strategy in this game is known to be PSPACE-complete.
  There are a few studies on algorithmic aspects of this problem.
  In this paper, we consider the problem from the viewpoint of fixed-parameter tractability. We show that the problem is fixed-parameter tractable parameterized by the size of a minimum vertex cover or the modular-width of a given graph.
  Moreover, we give a polynomial kernelization with respect to neighborhood diversity.
\end{abstract}
%
%
%
%=====================================================
\section{Introduction}\label{sec:intro}
{\em Kayles} is a two-player game with bowling pins and a ball.
In this game, two players alternately roll a ball down towards a row of pins.
Each player knocks down either a pin or two adjacent pins in their turn.
The player who knocks down the last pin wins the game.
This game has been studied in combinatorial game theory and the winning player can be characterized in the number of pins at the start of the game.

Schaefer \cite{Sch78} introduced a variant of this game on graphs, which is known as {\em Node Kayles}.
In this game, given an undirected graph, two players alternately choose a vertex, and the chosen vertex and its neighborhood are removed from the graph. The game proceeds as long as the graph has at least one vertex and ends when no vertex is left. The last player wins the game as well as the original game.
He studied the computational complexity of this game.
In this context, the goal is to decide whether the first player can win the game even though the opponent optimally plays the game, that is, the first player has a winning strategy.
He showed that this problem (hereinafter referred to simply as {\sc Node Kayles}) is PSPACE-complete.

After this hardness result was shown, there are a few studies on algorithmic aspects of {\sc Node Kayles}.
Bodlaender and Kratsch \cite{BK02} proved that {\sc Node Kayles} can be solved in $O(n^{k + 2})$ time, where $n$ is the number of vertices and $k$ is the asteroidal number of the input graph. This implies several tractability results for some graph classes, including AT-free graphs, circular-arc graphs, cocomparability graphs, and cographs since every graph in these classes has a constant asteroidal number.
Fleischer and Trippen~\cite{FT04} gave a polynomial-time algorithm for star graphs with arbitrary hair length.  
Bodlaender et al. \cite{BKT15} studied {\sc Node Kayles} from the perspective of exact exponential-time algorithms and gave an algorithm that runs in time $O(1.6031^n)$.

In this paper, we analyze {\sc Node Kayles} from the viewpoint of parameterized complexity.
Here, we consider a {\em parameterized problem} with instance size $n$ and parameter $k$.
If there is an algorithm that runs in $f(k)n^{O(1)}$ time, where the function $f$ is computable and does not depend on the instance size $n$, the problem is said to be {\em fixed-parameter tractable}.
There are several possible parameterizations on {\sc Node Kayles}. 
One of the most natural parameterizations is that the number of turns: The problem asks if the first player can win the game within $k$ turns.
This problem is known as {\sc Short Node Kayles} and, however, known to be AW[$*$]-complete \cite{ADF95}, which means that it is unlikely to be fixed-parameter tractable.

For tractable parameterizations, we leverage {\em structural parameterizations}, meaning that we use the parameters measuring the complexity of graphs rather than that of the problem itself.
{\em Treewidth} is one of the most prominent structural parameterizations for hard graph problems.
This parameter measures the ``tree-likeness'' of graphs. In particular, a connected graph has treewidth at most one if and only if it is a tree.
Although we know that many graph problems are fixed-parameter tractable when parameterized by the treewidth of the input graph \cite{Bod93}, the computational complexity of {\sc Node Kayles} is still open even when the input graph is restricted to trees. 

In this paper, we consider three structural parameterizations.
We show that {\sc Node Kayles} is fixed-parameter tractable parameterized by {\em vertex cover number} or by {\em modular-width}.
More specifically, we show that {\sc Node Kayles} can be solved in $3^{\tau(G)}n^{O(1)}$ time or $1.6031^{\mw(G)}n^{O(1)}$ time, where $\tau(G)$ is the size of a minimum vertex cover of $G$ and $\mw(G)$ is the modular-width of $G$.
Moreover, we show that {\sc Node Kayles} admits a $2\cdot\nd(G)$-vertex kernel, where $\nd(G)$ is the neighborhood diversity of $G$.
To the best of author's knowledge, these are the first non-trivial results of the fixed-parameter tractability of {\sc Node Kayles}.

The algorithm we used in this paper is, in fact, identical to that of Bodlaender et al.~\cite{BKT15}.
They gave a simple dynamic programming algorithm to solve {\sc Node Kayles} with the aid of the famous notion {\em nimber} in combinatorial game theory. They showed that this dynamic programming runs in time proportional to the number of specific combinatorial objects, which they call {\em K-sets}, and gave an exponential upper bound on the number of K-sets in a graph $G$.
In this paper, we prove that the number of K-sets of $G$ is upper bounded by some functions in $\tau(G)$ or $\mw(G)$, which directly yields our claims.
We also note that our combinatorial analysis is highly stimulated by the work of \cite{FLMT18} for the number of minimal separators and potential maximal cliques of graphs.

%=====================================================
\section{Preliminaries}\label{sec:preli}

All graphs appearing in this paper are simple and undirected.
Let $G = (V, E)$ be a graph.
We denote by $N_G(v)$ the open neighborhood of $v \in V$ in $G$, that is, $N_G(v) = \{w \in V : \{v, w\} \in E\}$,
and by $N_G[v]$ the closed neighborhood of $v$ in $G$, that is, $N_G[v] = N_G(v) \cup \{v\}$.
Let $X \subseteq V$. We use $G[X]$ to denote the subgraph of $G$ induced by $X$.
We also use the following notations: $N_G(X) = \bigcup_{v \in X} N_G(v) \setminus X$ and $N_G[X] = \bigcup_{v \in X} N_G[v]$.
For disjoint subsets $X, Y \subseteq V$, we denote by $E(X, Y)$ the set of edges having one end point in $X$ and the other end point in $Y$.

\subsection{Sprague-Grundy Theory}\label{ssec:CGT}
The Sprague-Grundy theory provides unified tools to analyze many two-players impartial combinatorial games.
The central idea in this theory is to use {\em nimber}, which is a non-negative integer assigned to each position (or state) of a game.
The nimber of a position can be defined as follows.
Hereafter, we consider {\em normal play games}, that is, the player who makes the last move wins the game. 
If there is no move from the current position $p$, the nimber $\nim(p)$ of $p$ is defined to be zero.
Otherwise, the nimber of $p$ is inductively defined as: $\nim(p) = \mex(\{\nim(p_i) : 1 \le i \le m\})$, where $p_1, p_2, \ldots, p_m$ are the positions that can be reached from $p$ with exactly one move and $\mex(S)$ is the minimum non-negative integer not contained in $S$. 
We say that a position is called {\em a winning position} if the current player has a winning strategy from this position.
The following theorem characterizes winning positions of a game with respect to nimbers.
\begin{theorem}[\cite{Con76}]\label{thm:nimber}
  A position $p$ is a winning position if and only if $\nim(p) > 0$.
\end{theorem}

The Sprague-Grundy theory allows us not only a simple way to decide the winning player of a game but also an efficient way to compute those nimbers when a position of the game can be decomposed into two or more ``independent'' subpositions.
Consider two positions $p_1$ and $p_2$ of (possibly different) games.
Then, we can make a new position of the combined game in which each player chooses one of the two positions $p_1$ and $p_2$ and then moves the chosen position to a next position in each turn. When both games are over, so is the combined game. We denote by $p_1 + p_2$ the combined position of $p_1$ and $p_2$.

\begin{theorem}[\cite{Con76}]\label{thm:nimsum}
  For any two positions $p_1$ and $p_2$ of (possibly distinct) games, we have $\nim(p_1 + p_2) = \nim(p_1) \oplus \nim(p_2)$,
  where $\oplus$ is the bit-wise exclusive or of the binary representations of given numbers.
\end{theorem}

\subsection{K-set}\label{ssec:K-set}
Bodlaender et al. \cite{BKT15} introduced the notion of {\em K-set} to characterize each game position of Node Kayles.

\begin{definition}
  Let $G = (V, E)$.
  A {\em K-set} is a non-empty subset $W \subseteq V$ such that 
  \begin{itemize}
    \item $G[W]$ is connected and
    \item there is an independent set $X \subseteq V$ with $W = V \setminus N_G[X]$.
  \end{itemize}
  We call such a triple $(W, N_G(X), X)$ a {\em K-set triple} of $G$.
\end{definition}

Let us note that a K-set triple partitions $V$ into three sets. Moreover, there are no edges between $W$ and $X$ as $N_G(X)$ separates $X$ from $W$.

They analyzed Node Kayles via the Sprague-Grundy theorem.
In Node Kayles, each position corresponds to some induced graph of $G$ and a single move corresponds to choosing a vertex and then deleting it together with its neighbors.
If $G$ has two or more connected components $G_1, G_2, \ldots, G_k$, by Theorem~\ref{thm:nimsum},
the nimber $\nim(G)$ can be computed as $\nim(G_1) \oplus \nim(G_2) \oplus \cdots \oplus \nim(G_k)$.
This means that we can independently compute the nimber of each connected component.
An important consequence from this fact is that each position, of which we essentially need to compute the nimber, is a subgraph induced by some K-set with chosen vertices $X$.
Thus, a standard recursive dynamic programming algorithm over all K-sets, described in Algorithm~\ref{alg:dp}, solves {\sc Node Kayles}.
\begin{algorithm}
  \caption{A recursive algorithm for computing $\nim(G)$.}\label{alg:dp}
\begin{algorithmic}
  \Procedure{$\nim$}{$G = (V, E)$}
  \IIf{$G$ is empty} \Return 0 \EndIIf
  \If{$G$ has two or more connected components $G_1, G_2, \ldots, G_k$}
    \State \Return $\nim(G_1) \oplus \nim(G_2) \oplus \cdots \oplus \nim(G_k)$
  \EndIf
  \If{$\nim(G)$ has been already computed}
    \State \Return $\nim(G)$
  \EndIf
  \State \Return $\mex(\{\nim(G[V \setminus N_G[v]) : v \in V\})$
  \EndProcedure
\end{algorithmic}
\end{algorithm}

We denote by $\kappa(G)$ the number of K-sets of $G$.
\begin{lemma}[\cite{BKT15}]\label{lem:K-set}
  Let $G$ be an input graph with $n$ vertices. Then, {\sc Node Kayles} can be solved in $\kappa(G)n^{O(1)}$ time.
\end{lemma}

Bodlaender et al. also gave the following upper bound on the number of K-sets.

\begin{lemma}[\cite{BKT15}]\label{lem:general-ub}
  Let $G$ be a graph with $n$ vertices.
  Then, $\kappa(G) = O((1.6031-\varepsilon)^n)$ for some small constant $\varepsilon > 0$.
\end{lemma}

These lemmas immediately give an $O(1.6031^n)$-time algorithm for {\sc Node Kayles}. This raises a natural question: How large is the number of K-sets in general graphs? They also gave the following lower bound example. 
\begin{lemma}[\cite{BKT15}]\label{lem:general-lb}
  There is a graph $G$ of $n$ vertices satisfying
  $\kappa(G) = 3^{(n - 1) / 3} + 4(n - 1) / 3$.
\end{lemma}

Note that $3^{n/3} = \omega(1.4422^n)$ and hence there is still a gap between upper and lower bounds on the number of K-sets.
Figure~\ref{fig:lb} illustrates a lower bound example in Lemma~\ref{lem:general-lb}.
There are four K-sets $\{v^i_1\}, \{v^i_3\}, \{v^i_2, v^i_3\}, \{v^i_1, v^i_2, v^i_3\}$ not containing the root $r$ for each $1 \le i \le (n - 1) / 3$. Let $P_i = \{v^i_1, v^i_2, v^i_3\}$.
For any K-set $W$ containing $r$, there are three possibilities: $W \cap P_i = \emptyset$, $W \cap P_i = \{v^i_1\}$, or $W \cap P_i = P_i$ for each $1 \le i \le (n - 1) / 3$.
Therefore, there are exactly $3^{(n - 1)/3}$ K-sets containing $r$.

\begin{figure}
  \centering
  \includegraphics[width=4.5cm]{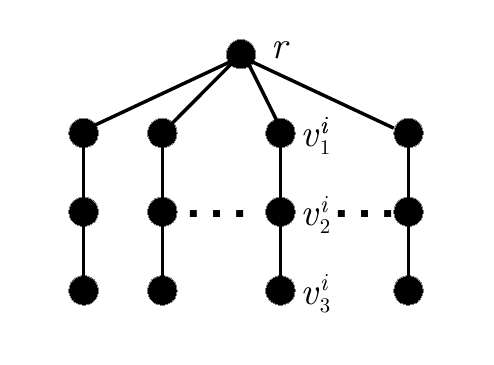}
  \caption{The lower bound example of Lemma~\ref{lem:general-lb}.}\label{fig:lb}
\end{figure}

%=====================================================
\section{Vertex Cover Number and K-Sets}\label{sec:vc}
In this section, we give an upper bound on the number of K-sets with respect to the minimum size of a vertex cover of the input graph $G$.
Let $\tau(G)$ be the size of a minimum vertex cover of $G$.

\begin{theorem}\label{thm:vc-ub}
  $\kappa(G) \le 3^{\tau(G)} + |V| - \tau(G) - 2^{\tau(G)}$.
\end{theorem}

\begin{proof}
  Let $C$ be a vertex cover of $G$ whose size is equal to $\tau(G)$.
  We say that a tuple $(X, Y, Z)$ is an {\em ordered tripartition} of $C$ if
  $X$, $Y$, and $Z$ are (possibly empty) disjoint subsets of $C$ such that
  $X \cup Y \cup Z = C$.
  From a K-set triple $(W, N_G(X), X)$, we can define an ordered tripartition of $C$: $(W \cap C, N_G(X) \cap C, X \cap C)$.
  Conversely, we consider how many K-set triples $(W, N_G(X), X)$ can be obtained from a fixed ordered tripartition $(W_C, N_C, X_C)$ of $C$, such that $W \cap C = W_C$, $N_G(X) \cap C = N_C$, and $X \cap C = X_C$.
  Suppose first that $W_C$ is empty. In this case, every K-set $W$ consists of exactly one vertex, which is in $V \setminus C$, and there are only $|V \setminus C|$ possibilities for such K-sets.

  Now, suppose that $W_C$ is not empty. We prove that a K-set triple $(W, N_G(X), X)$ is uniquely determined (if it exists) in this case. Obviously, we have $W_C \subseteq W$, $N_C \subseteq N_G(X)$, and $ X_C \subseteq X$ for any target K-set triple $(W, N_G(X), X)$.
  Consider a vertex $v \in V \setminus C$ in the independent set.
  Suppose that $v$ has a neighbor in $X_C$. Then, we conclude that $v$ belongs to $N_G(X)$ since
  $X$ is an independent set and there are no edges between $W$ and $X$.
  Suppose next that $v$ has a neighbor in $W_C$ and has no neighbor in $X_C$. 
  Then, $v$ must belong to $W$ since there are no edges between $W$ and $X$ and every vertex in $N_G(X)$ has a neighbor in $X$.
  Suppose otherwise. We conclude that $v$ belongs to $X$ since $G[W]$ must be connected and $v$ has no neighbor in $X$.

  Therefore, for each ordered tripartition $(W_C, N_C, X_C)$ of $C$, there is at most one K-set triple $(W, N_G(X), X)$ such that $W \cap C = W_C$, $N_G(X) \cap C = N_C$, and $X \cap C = X_C$, except for the case $W_C = \emptyset$.
  Clearly, the number of ordered tripartitions $(W_C, N_C, X_C)$ of $C$ with $W_C \neq \emptyset$ is $3^{\tau(G)} - 2^{\tau(G)}$ (as the number of ordered tripartitions with $W_C = \emptyset$ is $2^{\tau(G)}$) and hence the theorem follows.
\end{proof}

It is natural to ask whether the upper bound in Theorem~\ref{thm:vc-ub} can be improved.
However, this is essentially impossible since the lower bound example in Lemma~\ref{lem:general-lb} 
has a minimum vertex cover of size $(n - 1) / 3 + 1$, which implies our upper bound is tight up to a constant multiplicative factor.

%=====================================================
\section{Modular-Width and K-sets}\label{sec:modular}
Let $G = (V, E)$ be a graph. A vertex set $M \subseteq V$ is a {\em module} of $G$ if for every $v \in V \setminus M$, either $M \cap N_G(v) = \emptyset$ or $M \subseteq N_G(v)$ holds.
A {\em modular decomposition} of $G$ is a recursive decomposition of $G$ into some modules.
Here, we define a width parameter associated with this decomposition.

\begin{definition}\label{def:md}
  The {\em modular-width} of $G$, denoted by $\mw(G)$, is the minimum integer $k \ge 1$ such that at least one of the following conditions hold:
  \begin{itemize}
    \item[M1] $G$ consists of a single vertex or
    \item[M2] $V$ can be partitioned into at most $k$ modules $V_1, V_2, \ldots, V_{k'}$ such that
    the modular-width of each $G[V_i]$ is at most $k$.
  \end{itemize}
\end{definition}

In several papers, the definition of modular-width is slightly different from ours. More specifically, the following two conditions are included in addition to those in Definition~\ref{def:md}:
\begin{itemize}
  \item[M3] $G$ is a disjoint union of two graphs of modular-width at most $k$,
  \item[M4] $G$ is a join of two graphs of modular-width at most $k$, that is, $G$ is obtained by taking a disjoint union of the two graphs and adding edges between those graphs.
\end{itemize}

Note that the conditions M3 and M4 are special cases of M2, and this difference may change the modular-width of some extreme cases.
However, we emphasize that this does not change our discussion and then use the simpler version for expository purposes.
The modular-width and its decomposition can be computed in linear time \cite{MS94}.

Suppose that $G$ has modular-width at most $k$.
Then, $V$ can be partitioned into at most $k$ modules $V_1, V_2, \ldots, V_{k'}$ with $k' \le k$, such that $\mw(G[V_i]) \le k$ for $1 \le i \le k'$.
Observe that for every pair of distinct modules $V_i$ and $V_j$, 
either $E(V_i, V_j) = \emptyset$ or $E(V_i, V_j) = V_i \times V_j$.
To see this, consider an arbitrary $v \in V_i$.
Since $V_j$ is a module in $G$, either $V_j \cap N_G(v) = \emptyset$ or $V_j \cap N_G(v) = V_j$.
If $V_j \cap N_G(v) = \emptyset$, $V_j \cap N_G(w) = \emptyset$ holds for every $w \in V_i$ since $V_i$ is also a module in $G$.
Otherwise, $V_j \cap N_G(w) = V_j$ holds for every $w \in V_i$.

Let $H$ be a graph with vertex set $\{v_1, v_2, \ldots, v_{k'}\}$ such that
$v_i$ is adjacent to $v_j$ if and only if $E(V_i, V_j) = V_i \times V_j$.
In other words, $H$ is obtained from $G$ by identifying each module $V_i$ into a single vertex $v_i$.
We say that two modules $V_i$ and $V_j$ are adjacent if $v_i$ and $v_j$ are adjacent in $H$.
When $H$ is obtained as above, $G$ is called an {\em expansion graph} of $H$ and, for a subset $U \subseteq \{v_1, v_2, \ldots, v_{k'}\}$, the vertex set $\bigcup_{v_i \in U} V_i$ is an {\em expansion} of $U$.

\begin{lemma}\label{lem:expansion}
  Let $W$ be a K-set of $G$ with K-set triple $(W, N_G(X), X)$.
  Let $\{V_1, V_2, \ldots, V_{k'}\}$ is a set of modules of $G$ that partitions $V$.
  Let $H$ be a graph with vertex set $\{v_1, v_2, \ldots, v_{k'}\}$ whose expansion graph is $G$.
  Then, at least one of the following conditions hold:
  \begin{itemize}
    \item $W$ is an expansion of some K-set of $H$ or
    \item $W$ is a K-set of $G[V_i]$ for some $1 \le i \le k'$.
  \end{itemize}
 \end{lemma}
 \begin{proof}
  Observe first that if $G$ has two or more connected components, each module is included in some connected component of $G$.
  Moreover, as $W$ is connected, $W$ is also a K-set of a connected component of $G$.
  Therefore, in this case, we inductively apply the lemma to this connected component.
  Thus, we consider the case where $G$ is connected.

  Let $(W, N_G(X), X)$ be a K-set triple of $G$.
  In the following, we prove that $W$ is a K-set of $G[V_i]$ for some $1 \le i \le k'$ under the assumption that $W$ is not an expansion of any K-sets of $H$.
  
  We first observe that there is a module $V_i$ with $V_i \cap W \neq \emptyset$ and $V_i \setminus W \neq \emptyset$.
  To see this, suppose to the contrary that every module $V_i$ satisfies either $V_i \cap W = \emptyset$ or $V_i \subseteq W$.
  Let $W_H$ and $X_H$ be the sets of vertices of $H$ whose corresponding modules have an non-empty intersection with $W$ and $X$, respectively.
  We claim that $(W_H, N_H(X_H), X_H)$ is a K-set
  As $G[W]$ is connected, $H[W_H]$ is connected as well.
  Similarly, $X_H$ is an independent set of $H$.
  By the assumption that either $V_i \cap W = \emptyset$ or $V_i \subseteq W$ for every module $V_i$, we have $W_H \cap X_H = \emptyset$.
  If there is an edge between a vertex in $W_H$ and a vertex in $X_H$, then there are adjacent modules $V_i$ and $V_j$ such that $V_i$ has a vertex in $W$ and $V_j$ has a vertex in $X$, which contradicts to the fact that $W = N_G[X]$.
  Thus, $W_H = N_H[X_H]$ and therefore $W_H$ is a K-set of $H$.
  
  Let $V_i$ be a module of $G$ with $V_i \cap W \neq \emptyset$ and $V_i \setminus W \neq \emptyset$.
  We then claim that $W \subseteq V_i$.
  Suppose for contradiction that there is a module $V_j$ distinct from $V_i$ such that $V_j \cap W \neq \emptyset$.
  As $W$ is connected in $G$, we can choose $V_j$ so that $V_j$ is adjacent to $V_i$.
  If $V_i \cap X \neq \emptyset$, we have $V_j \subseteq N_G(X)$, which contradicts the fact that $V_j$ has a vertex of $W$.
  Moreover, if $V_i \cap N_G(X) \neq \emptyset$, there is an adjacent module $V_k$ with $V_k \cap X \neq \emptyset$ as $V_i \cap X = \emptyset$.
  This implies that $V_i$ is entirely contained in $N_G(X)$, which also leads to a contradiction.
  Therefore, we have $W \subseteq V_i$.

  Clearly, $W$ is connected in $V_i$, and $V_i \cap X$ is also an independent set in $G[V_i]$ with $W = V_i \setminus N_{G[V_i]}(V_i \cap X)$. Therefore, $W$ is a K-set in $G[V_i]$. 
 \end{proof}
 
 For a positive integer $n$, a {\em partition of $n$} is a multiset of positive integers $n_1, n_2, \ldots, n_t$ with $\sum_{1 \le i \le t} n_i = n$.
 \begin{lemma}\label{lem:eps}
    For every $\varepsilon > 0$, there is a constant $n_\varepsilon$ such that for any integer $n \ge n_{\varepsilon}$ and for any partition $n_1, n_2, \ldots, n_t$ of $n$ with $t \ge 2$, it holds that $\sum_{1 \le i \le t} n_i^{1+\varepsilon} + 1 \le n^{1+\varepsilon}$.
 \end{lemma}
 \begin{proof}
    Assume that $t = 2$ and $n_1 + n_2 = n$ for some integers $n_1, n_2$, where $n_1 \ge n_2 \ge 1$.
    From the convexity of the function $f(x) = x^{1 + \varepsilon}$, we have $1 + (n - 1)^{1 + \varepsilon} \ge n_1^{1+\varepsilon} + n_2^{1+\varepsilon}$.
    Thus, it suffices to prove the lemma for the case $n_1 = n - 1$ and $n_2 = 1$.
    
    Since $f$ is convex and differentiable, by an equivalent characterization of differentiable convex functions~\cite{Boyd:Convex:2004}, we have $f(x) \ge f(y) + f'(y)(y - x)$ for $x, y \in \mathbb R$, which implies
    \begin{align*}
        f(x) = x^{1 + \varepsilon} \ge y^{1 + \varepsilon} + (1 + \varepsilon)y^{\varepsilon}(x - y)
    \end{align*}
    for $x, y \in \mathbb R$.
    Now, we set $x = 1 + 1 / (n - 1)$ and $y = 1$. Then, it follows that $(1 + 1 / (n - 1))^{1 + \varepsilon} \ge 1 + (1 + \varepsilon)/(n - 1)$.
    Thus, we have
    \begin{eqnarray*}
        n^{1 + \varepsilon} &=& (n - 1)^{1 + \varepsilon}(1 + \frac{1}{n - 1})^{1 + \varepsilon}\\
        &\ge& (n - 1)^{1 + \varepsilon}(1 + \frac{1 + \varepsilon}{n - 1})\\
        &=& (n - 1)^{1 + \varepsilon} + (1 + \varepsilon)(n - 1)^{\varepsilon}.
    \end{eqnarray*}
    For every $n \ge (\frac{2}{1 + \varepsilon})^{1 / \varepsilon} + 1$, it holds that $(1 + \varepsilon)(n - 1)^{\varepsilon} \ge 2$.
    Therefore, the lemma holds for $t = 2$. By inductively applying this argument for $t > 2$, the lemma follows.
 \end{proof}

  Now, we are ready to prove the main claim of this section.
 \begin{theorem}\label{thm:mw-ub}
  For every graph $G$ and every $\varepsilon > 0$, $\kappa(G) = O(1.6031^{\mw(G)}|V|^{1 + \varepsilon})$.
 \end{theorem}
 \begin{proof}
  The proof is by induction on the number of vertices. The base case, where $G$ consists of a single vertex, is trivial.
  Let $V_1, V_2, \ldots, V_k$ be $k \le \mw(G)$ modules of $G$ such that $\mw(G[V_i]) \le \mw(G)$ for each $1 \le i \le k$, and let $H$ be a graph of $k$ vertices whose expansion graph is $G$.
  By Lemma~\ref{lem:expansion}, every K-set of $G$ is either an expansion of a K-set of $H$ or a K-set of some subgraph $G[V_i]$.
  Let $\CW_0$ be the set of K-sets, each of which is an expansion of a K-set of $H$ and let $\CW_i$ be the set of K-sets of $G[V_i]$ for $1 \le i \le k$.
  By Lemma~\ref{lem:general-ub}, there is some constant $c > 0$ such that $|\CW_0| = \kappa(H) \le c \cdot 1.6031^{k}$. We apply the induction hypothesis to each $G[V_i]$, and then we have $|\CW_i| = \kappa(G[V_i]) \le c' \cdot 1.6031^{\mw(G)}|V_i|^{1 + \varepsilon}$ for some constant $c' > 0$.
  Summing up of all $|\CW_i|$, we have
  \begin{eqnarray*}
    \kappa(G) &\le& \sum_{0 \le i \le k} |\CW_i| \\
    &\le& \max(c, c') \cdot 1.6031^{\mw(G)}(1 + \sum_{1 \le i \le k}|V_i|^{1 + \varepsilon}) \\
    &=& O(1.6031^{\mw(G)}|V|^{1 + \varepsilon}).
  \end{eqnarray*} 
  The last equality follows from Lemma~\ref{lem:eps}.
 \end{proof}
 
Let us note that this argument also works on trees.
Bodlaender et al. \cite{BKT15} proved that $\kappa(T) \le n\cdot 3^{n/3}$ for any tree $T$ with $n$ vertices.

\begin{lemma}
    Let $T$ be a tree and let $V_1, V_2, \ldots, V_k$ be a set of modules of $T$ that partitions the vertex set of $T$.
    Then, a graph $H$ whose expansion graph is $T$ has no cycles.
\end{lemma}

\begin{proof}
    Suppose that $H$ has a cycle $C$. Obviously there are at least three modules.
    If every module that is involved by this cycle has exactly one vertex, then $C$ is also a cycle of $T$, a contradiction.
    Therefore, we assume that there is a vertex $v_i$ on the cycle whose module $V_i$ contains at least two vertices.
    Let $v_b$ and $v_f$ be the vertices adjacent to $v_i$ on $C$ and let $V_b$ and $V_f$ be the corresponding modules.
    Then, $x \in V_b$, $y \in V_f$, and any pair of vertices in $V_i$ form a cycle, which contradicting to the fact that $T$ is a tree.
\end{proof}

Therefore, we can apply the improved upper bound for trees to the argument in Lemma~\ref{lem:expansion}.
Hence, we have $\kappa(T) = O(3^{\mw(T)/3}n^{1 + \varepsilon})$.

\begin{corollary}
    For every tree $T$ and every $\varepsilon > 0$, $\kappa(T) = O(3^{\mw(T)/3}n^{1 + \varepsilon})$,
    where $n$ is the number of vertices in $T$.
\end{corollary}

%=====================================================
\section{Polynomial Kernel with Neighborhood Diversity}
In the previous section, we use the modular-width of a graph, which measure how well the graph can recursively decomposed into modules.
In this section, we use another graph parameter related to modular-width.

Let $M$ be a module of $G$.
We call M a {\em clique module} if $M$ is a clique of $G$ and an {\em independent module} if $M$ is an independent set of $G$.
The {\em neighborhood diversity} of a graph $G$, denoted by $\nd(G)$, is the minimum integer $k$ such that $G$ can be partitioned into $k$ modules, each of which is either a clique module or an independent module.
Since a single vertex is a clique module (or an independent module), $\nd(G) \le |V|$.
The neighborhood diversity of a graph can be computed in linear time \cite{Lam12}.

In this section, we prove that for every graph $G$, there is a graph $H$ with $\nim(G) = \nim(H)$ such that $H$ has at most $2\cdot\nd(G)$ vertices.
Moreover, such a graph $H$ can be computed in linear time, that is, {\sc Node Kayles} admits a polynomial kernelization with respect to neighborhood diversity. 

Let $M$ be a clique module in $G$.
A crucial observation is that for any pair of vertices $u, v \in M$, $N_G[u] = N_G[v]$.
Moreover, since $M$ is a module, for every $v \in V \setminus M$, either $M \subseteq N_G[v]$ or $M \cap N_G[v] = \emptyset$ holds.
This means that every vertex in $M$ is indistinguishable.
The idea is formalized as follows.

\begin{lemma}\label{lem:clique-module}
    Let $M$ be a clique module of size at least three in $G$ and let $H = (V_H, E_H)$ be the graph obtained from $G$ by identifying $M$ into a single vertex $v_M$.
    Then, $\nim(G) = \nim(H)$.
\end{lemma}

\begin{proof}
    We prove the lemma by induction on the number of vertices of $V \setminus M$.
    If $V = M$, the lemma is trivial.
    
    Suppose that $G$ contains at least one vertex other than $M$.
    For every $v \in M$, we have $G[V \setminus N_G[v]] = H[V_H \setminus N_H[v_M]]$.
    Let $v \in V \setminus M$.
    If $M \subseteq N_G(v)$, then we also have $G[V - N_G[v]] = H[V_H \setminus N_H[v_M]]$.
    Otherwise, let $G_v = G[V \setminus N_G[v]]$. 
    Since $M$ is also an clique module in $G_v$, by the induction hypothesis, we have $\nim(G_v) = \nim(H_v)$, where $H_v$ is the graph obtained from $G_v$ by identifying the vertices $M$ into a single vertex.
    Therefore, $\nim(G) = \mex(\{\nim(H[V_H \setminus N_H[v_M]])\} \cup \{\nim(H_v) : v \in V \setminus M\}\})$.
    Since $H_v = H[V_H \setminus N_H[v]]$, we have 
    \[
    \nim(G) = \mex(\{\nim(H[V_H \setminus N_H[v]]) : v \in V_H\}) = \nim(H).
    \]
    This completes the proof.
\end{proof}

Now we can assume that $G$ has no clique modules of size at least two.
Let $M$ be an independent module of size at least two.
Let us observe that independent modules are slightly different from clique modules: If a vertex $v$ in $M$ is chosen,
then $G[V \setminus N_G[v]]$ still contains the other vertex of $M$.
The key to handle these remaining vertices is that $M \setminus \{v\}$ are isolated vertices in $G[V \setminus N_G[v]]$.
If $G$ contains two isolated vertices $u$ and $v$, by Theorem~\ref{thm:nimsum}, we have 
\[
\nim(G) = \nim(G[V \setminus \{u, v\}]) \oplus \nim(G[\{u\}]) \oplus \min(G[\{v\}]) = \nim(G[V \setminus \{u, v\}]).
\]
This implies, that $\nim(G[V \setminus N_G[v]]) = \nim(G[V \setminus (N_G[v] \cup M)]) \oplus p(|M| - 1)$, where $p(k) = 1$ if $k$ is odd and $p(k) = 0$ if $k$ is even.
Since the right-hand side of this equality does not depend on a specific vertex in $M$, does depend on the parity of $|M|$.

\begin{lemma}\label{lem:independent-module}
    Let $M$ be an independent module of size at least three in $G$ and let $H = (V_H, E_H)$ be the graph obtained from $G$ by arbitrary removing $|M| - 2$ vertices of $M$.
    Then, $\nim(G) = \nim(H)$.
\end{lemma}
\begin{proof}
    We prove the lemma by induction on the number of vertices of $V \setminus M$.
    If $V = M$, by Theorem~\ref{thm:nimsum}, $\nim(G) = p(|V|) = \nim(H)$ since $|V| \equiv |V_H| \bmod{2}$.
    
    Suppose that $V \setminus M$ contains at least one vertex.
    For $v \in V \setminus M$, as $M$ is a module, either $M \subseteq N_G(v)$ or $M \cap N_G(v) = \emptyset$ holds.
    Suppose $v \in V \setminus M$. This case is almost the same as in Lemma~\ref{lem:clique-module}.
    If $M \subseteq N_G[v]$, then we have $G[V \setminus N_G[v]] = H[V_H \setminus N_H[v]]$.
    Otherwise, let $G_v = G[V \setminus N_G[v]]$. 
    By the induction hypothesis, we have $\nim(G_v) = \nim(H_v)$, where $H_v$ is obtained from $G_v$ by removing $|M| - 2$ vertices of $M$.
    Suppose  $v \in M$. Then,
    \begin{eqnarray*}
    \nim(G[V \setminus N_G[v]]) &=& \nim(G[V \setminus (N_G[v] \cup M)]) \oplus p(|M| - 1)\\
    &=& \nim(H[V_H \setminus N_H[v] \cup M]) \oplus p(|M| - 1)\\
    &=& \nim(H[V_H \setminus N_H[v]]).
    \end{eqnarray*}
    The third equality follows from the fact that $M$ is an independent set in $H[V_H \setminus N_H[v]]$.
    Therefore,
    \begin{eqnarray*}
        \nim(G) &=& \mex(\{\nim(G[V \setminus N_G[v]]) : v \in V\})\\
        &=& \mex(\{\nim(H[V_H \setminus N_H[v]]) : v \in V \})\\
        &=& \nim(H),
    \end{eqnarray*}
    which completes the proof.
\end{proof}

Our kernelization algorithm is straightforward. 
For each clique module $M$, remove all but one vertex in $M$, and for each independent module $M$, remove $|M| - 2$ vertices in $M$ from $G$.
Let $H$ be the resulting graph.
By Lemmas~\ref{lem:clique-module} and \ref{lem:independent-module}, $\nim(G) = \nim(H)$.
Moreover, $H$ contains at most $2 \cdot \nd(G)$ vertices.

%=====================================================
\section{Concluding Remarks}\label{sec:concl}
In this paper, we give a new running time analysis of the known algorithm for {\sc Node Kayles}.
Bodlaender et al. \cite{BKT15} showed that {\sc Node Kayles} can be solved in $\kappa(G)n^{O(1)}$ time, 
where $\kappa(G)$ is the number of K-sets in $G$ and that $\kappa(G) = O(1.6031^n)$.
We analyze the number of K-sets from the perspective of structural parameterizations of graphs, and
show that $\kappa(G) \le 3^{\tau(G)} + n - \tau(G) - 2^{\tau(G)}$ and $\kappa(G) = O(1.6031^{\mw(G)}n^{1 + \varepsilon})$ for every constant $\varepsilon > 0$.
The first upper bound is tight up to the constant factor, and the second one improves
the known upper bound due to Bodlaender et al.~\cite{BKT15} when the modular-width is relatively small compared to the number of vertices.
We also give a polynomial kernelization for {\sc Node Kayles} with respect to the neighborhood diversity of a graph.

It would be interesting to know whether other graph parameters yield a new upper bound on the number of K-sets.
However, the lower bound example in Lemma~\ref{lem:general-lb} indicates some limitation on this question.
In particular, the treewidth, pathwidth, and even treedepth of this example are all bounded.
Moreover, we can transform the lower bound example into a bounded degree tree as in Fig.~\ref{fig:lb2}, which has also exponentially many K-sets.
Note that this argument does not imply any particular complexity result of {\sc Node Kayles} on trees.
\begin{figure}
  \centering
  \includegraphics[width=3.7cm]{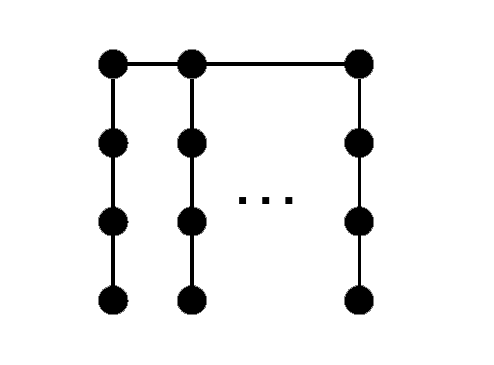}
  \caption{Similarly to the example in Lemma~\ref{lem:general-lb}, this graph has exponentially many K-sets.}\label{fig:lb2}
\end{figure}

\section*{Acknowledgments}
 The author thanks Kensuke Kojima for simplifying the proof of Lemma~\ref{lem:eps} and anonymous reviewers for valuable comments.
 This work was partially supported by JST CREST JPMJCR1401. 
\bibliographystyle{plain}
\bibliography{references}

\begin{thebibliography}{10}

\bibitem{ADF95}
Karl~R. Abrahamson, Rodney~G. Downey, and Michael~R. Fellows.
\newblock Fixed-parameter tractability and completeness {IV:} on completeness
  for {W[P]} and {PSPACE} analogues.
\newblock {\em Ann. Pure Appl. Log.}, 73(3):235--276, 1995.

\bibitem{Bod93}
Hans~L. Bodlaender.
\newblock A tourist guide through treewidth.
\newblock {\em Acta Cybern.}, 11(1-2):1--21, 1993.

\bibitem{BK02}
Hans~L. Bodlaender and Dieter Kratsch.
\newblock Kayles and nimbers.
\newblock {\em J. Algorithms}, 43(1):106--119, 2002.

\bibitem{BKT15}
Hans~L. Bodlaender, Dieter Kratsch, and Sjoerd~T. Timmer.
\newblock Exact algorithms for kayles.
\newblock {\em Theor. Comput. Sci.}, 562:165--176, 2015.

\bibitem{Boyd:Convex:2004}
Stephen Boyd and Lieven Vandenberghe.
\newblock {\em Convex Optimization}.
\newblock {Cambridge University Press}, 2004.

\bibitem{Con76}
John~H. Conway.
\newblock {\em On numbers and games, Second Edition}.
\newblock A {K} Peters, 2001.

\bibitem{FT04}
Rudolf Fleischer and Gerhard Trippen.
\newblock Kayles on the way to the stars.
\newblock In {\em Computers and Games}, volume 3846 of {\em Lecture Notes in
  Computer Science}, pages 232--245. Springer, 2004.

\bibitem{FLMT18}
Fedor~V. Fomin, Mathieu Liedloff, Pedro Montealegre, and Ioan Todinca.
\newblock Algorithms parameterized by vertex cover and modular width, through
  potential maximal cliques.
\newblock {\em Algorithmica}, 80(4):1146--1169, 2018.

\bibitem{MS94}
Ross~M. McConnell and Jeremy~P. Spinrad.
\newblock Linear-time modular decomposition and efficient transitive
  orientation of comparability graphs.
\newblock In {\em {SODA}}, pages 536--545. {ACM/SIAM}, 1994.

\bibitem{Sch78}
Thomas~J. Schaefer.
\newblock On the complexity of some two-person perfect-information games.
\newblock {\em J. Comput. Syst. Sci.}, 16(2):185--225, 1978.

\end{thebibliography}

\end{document}